\documentclass[12pt,a4paper]{article}
\usepackage[top=20mm, left=25mm, right=25mm, bottom=20mm]{geometry}
\usepackage{graphicx} 
\usepackage[colorlinks = true]{hyperref}
\usepackage[utf8]{inputenc}
\usepackage{xcolor}
\usepackage{amsmath}
\usepackage{amsfonts}
\usepackage{amsthm}
\newtheorem{thm}{Theorem}[section]

\newtheorem{theorem}{Theorem}[section]
\newtheorem{lemma}[theorem]{Lemma}

\theoremstyle{definition}

\newtheorem{conjecture}[theorem]{Conjecture}

\title{Strengthening Wilf's lower bound on clique number\thanks{This work is partially supported by the Department of Science and Technology (Government of India) under SERB Project SRG/2022/002219, DST INSPIRE program (grant number Inspire 16/2020), and project IITI/YFRCG/2023-24/03.}}

\author{
    Hareshkumar Jadav, 
    Sreekara Madyastha,\\
    Rahul Raut,
    Ranveer Singh\\
    \small Department of Computer Science and Engineering,\\
    \small Indian Institute of Technology Indore, India
}

\date{} 

\begin{document}

\maketitle
\begin{abstract}
    Given an integer $k$, deciding whether a graph has a clique of size $k$ is an NP-complete problem. Wilf's inequality provides a spectral bound for the clique number of simple graphs. Wilf's inequality is stated as follows: $\frac{n}{n - \lambda_{1}} \leq \omega$, where $\lambda_1$ is the largest eigenvalue of the adjacency matrix $A(G)$, $n$ is the number of vertices in $G$, and $\omega$ is the clique number of $G$. Strengthening this bound, Elphick and Wocjan proposed a conjecture in 2018, which is stated as follows: $\frac{n}{n - \sqrt{s^{+}}} \leq \omega$, where $s^+ = \sum_{\lambda_{i} > 0} \lambda_{i}^2$ and $\lambda_i$  are the eigenvalues of $A(G)$. In this paper, we have settled this conjecture for some classes of graphs, such as conference graphs, strongly regular graphs with $\lambda = \mu$ (i.e., $srg(n, d, \mu, \mu)$) and $n\geq 2d$, the line graph of $K_{n}$, the Cartesian product of strongly regular graphs, and Ramanujan graph with $n\geq 11d$.
    
\end{abstract}



\section{Introduction}

Let $G = (V, E)$ be a simple graph with $n$ vertices and $m$ edges, where $V$ is the vertex set and $E$ is the edge set. Let $A(G)$ be the adjacency matrix of the graph $G$ and $\lambda_{1} \geq \lambda_{2} \geq \dots \geq \lambda_{n}$ be the eigenvalues of $A(G)$. The eigenvalues of $G$ are the eigenvalues of $A(G)$. In an undirected graph, a \textit{clique} is a subset of vertices such that every two distinct vertices in the subset are adjacent. The size of the maximum clique in the graph $G$ is called the \textit{clique number}, denoted as $\omega$ (or $\omega(G)$). Determining the clique number of a given graph is an NP-hard problem \cite{cormen}, meaning there is no polynomial-time algorithm to compute it (unless P $\neq$ NP). In 1986 \cite{wilf}, Wilf proved a spectral bound for the clique number. For graph $G$,
\[
     \frac{n}{n - \lambda_{1}}\leq \omega.
\]
Alternatively, it can be expressed as
\[
    \lambda_{1} \leq n \left(1 - \frac{1}{\omega}\right).
\]
To prove this, Wilf used the Motzkin-Straus theorem \cite{motzkin}, which is stated as follows. Let $S_n = \{ x \in \mathbb{R}^n : x \geq 0 \ \text{and} \ \textbf{1}^T x = 1 \}$ be a simplex in $\mathbb{R}^n$. Then 
\[
    \max_{x \in S_n} \sum_{(i, j) \in E} x_i x_j = \frac{1}{2} \left(1 - \frac{1}{\omega}\right).
\]

In 2018, Elphick and Wocjan \cite{elphick} proposed a conjecture:
\begin{conjecture} \label{conj2}
    Let $G$ be a graph and $s^+ = \sum_{\lambda_{i} > 0} \lambda_{i}^2$. Then 
\[
    \sqrt{s^{+}} \leq n \left(1 - \frac{1}{\omega} \right).
\]
Alternatively, it can be expressed as
\[
    \frac{n}{n - \sqrt{s^{+}}} \leq \omega.
\]
\end{conjecture}

Experimentally, they tested thousands of named graphs with up to 40 vertices but did not find any counterexamples. Additionally, they proved that the conjecture holds for triangle-free graphs, weakly perfect graphs, Kneser graphs, and almost all graphs (using the Erdős–Rényi random graph model with probability 0.5). This conjecture is also stated as the second problem in the survey on open problems in spectral graph theory in \cite{Liu}. 
In this paper, we have settled this conjecture for some classes of graphs, such as conference graphs, strongly regular graphs with $\lambda = \mu$ (i.e., $srg(n, d, \mu, \mu)$) and $n\geq 2d$, the line graph of $K_{n}$, the Cartesian product of strongly regular graphs, and Ramanujan graph with $n\geq 11d$.

The rest of the paper is organized as follows. In Section \ref{section_srg}, we prove that the conjecture holds for conference graphs and $srg(n, d, \lambda, \mu)$ with $\lambda = \mu$ and $n\geq2d$. In Section \ref{section_line}, we prove that the conjecture holds for the line graph of a complete graph $K_{n}$. In Section \ref{section_car}, we prove that if the graph $G$ is strongly regular and satisfies the conjecture, then the conjecture is true for the Cartesian product of $G$ with $G$, denoted as ($G\; \Box \;G$). 
In Section \ref{section_ramanujan}, we prove that the conjecture holds for Ramanujan graphs if $n\geq 11d$.

\section{Strongly regular graphs} \label{section_srg}

A \textit{strongly regular graph} is denoted by $srg(n, d, \lambda, \mu)$, where $n$ is the number of vertices, $d$ is the degree of each vertex, $\lambda$ is the number of common neighbors for adjacent vertices, and $\mu$ is the number of common neighbors for non-adjacent vertices. 
We use the standard notation $(d^1, r^f, s^g)$ to denote eigenvalues and their multiplicity, where $d$, $r$, and $s$ are eigenvalues and $1$, $f$, and $g$ denote their multiplicities, respectively. All three eigenvalues of the strongly regular graph are known, which are as follows:
\[
d \text{ with multiplicity } 1,
\]
\[
r = \frac{1}{2} \left[ (\lambda-\mu) + \sqrt{(\lambda-\mu)^2 + 4(d-\mu)} \right]
\]
with multiplicity
\[
f = \frac{1}{2}\left[ (n - 1) - \frac{2d + (n - 1)(\lambda - \mu)}{\sqrt{(\lambda - \mu)^2 + 4(d - \mu)}}\right],
\]
and
\[
s = \frac{1}{2} \left[ (\lambda-\mu) - \sqrt{(\lambda-\mu)^2 + 4(d-\mu)} \right]
\]
with multiplicity
\[
g = \frac{1}{2}\left[ (n - 1) + \frac{2d + (n - 1)(\lambda - \mu)}{\sqrt{(\lambda - \mu)^2 + 4(d - \mu)}}\right].
\]
We will use these formulas in this work. For more details on strongly regular graphs, refer to \cite{brouwer2022strongly,godsil}.

In this section, we first prove the conjecture for conference graphs \cite{brouwer}, which are strongly regular graphs with parameters \((4\mu + 1, 2\mu, \mu - 1, \mu)\) for some integer $\mu$.

\begin{thm}
    If a graph $G$ is conference graph, then \[\frac{n}{n - \sqrt{s^+}} \: \leq \: \omega(G).\]
\end{thm}

\begin{proof}
A conference graph can be described by the parameters 
$
(4\mu +1 ,2\mu,\mu-1,\mu).
$
Therefore, the positive eigenvalues and their corresponding multiplicities are,
\[
2\mu \text{ with multiplicity 1,}
\]
and
\[
\frac{1}{2}\left(-1+\sqrt{1+4\mu}\right) \text{with multiplicity } 2\mu.
\]
Therefore,
\[
s^+ =4\mu^2 + \frac{\mu}{2}\left(\sqrt{4\mu+1}-1\right)^2
    = 6\mu^2 +\mu -\mu\sqrt{1+4\mu}.
\]
So we can write,
\[
    \frac{n}{n-\sqrt{s^+}} = \frac{4\mu+1}{4\mu +1 - \sqrt{6\mu^2 +\mu -\mu\sqrt{1+4\mu}}} \leq \frac{4\mu+1}{4\mu +1 - \sqrt{6\mu^2 +\mu -2\mu\sqrt{\mu}}}.
\]
For $\mu > 1$, it implies that $\lambda \geq 1$, therefore we have at least one triangle in the graph. Hence

\[
\frac{n}{n-\sqrt{s^+}} \leq \frac{4\mu+1}{4\mu +1 - \sqrt{6\mu^2 +\mu -2\mu\sqrt{\mu}}} \leq 3.
\]

If $\mu = 1$, then $\lambda = 0$, which means that two adjacent vertices do not share any common neighbors, resulting in a triangle-free graph. The conjecture has already been proven for triangle-free graphs. This completes the proof.

\end{proof}

Now, we will prove that the conjecture is true for strongly regular graphs when $\lambda = \mu$ and $n\geq 2d$.

\begin{thm}
    If a graph $G$ is $srg(n, d, \lambda, \mu)$ with $\lambda=\mu$ and $n\geq 2d$, then \[\frac{n}{n - \sqrt{s^+}} \: \leq \: \omega(G).\]
\end{thm}

\begin{proof} 
For a given $srg(n,d,\lambda, \mu ) $ with $\lambda=\mu$,
\begin{equation}
\label{lambdaIsmu}
    s^+ = d^2 + \frac{1}{2}(d-\mu)\left( (n - 1) - \frac{d}{\sqrt{d - \mu}}\right).
\end{equation}
Parameters of $srg(n, d, \lambda, \mu)$ have following relation: 
\[
(n - d - 1)\mu = d(d-\lambda - 1)
\]
and for $\lambda = \mu$,
\[
\mu = \frac{d(d-1)}{n - 1}.
\]
By substituting the value of $\mu$ in (\ref{lambdaIsmu}), we get 
\[
s^+ = d^2 + \frac{1}{2}\left(d-\frac{d(d-1)}{n - 1}\right)\left( (n - 1) - \sqrt{\frac{d(n - 1)}{n - d}}\right)
\]
and
\[
\text{if } n\geq 2d\text{, then }s^+ \leq \frac{4n^2}{9}.
\]
If $s^+ \leq \frac{4n^2}{9}$ then,
\[
\frac{n}{n - \sqrt{s^+}} \leq 3.
\]
When the clique number is 2, the graph is triangle-free, and the conjecture is already proven for triangle-free graphs. This completes the proof.
\end{proof}

By this, we can say that for strongly regular graphs like the Gewirtz graph, i.e., $srg(56,10,2,2)$, the conjecture is true.

\section{Line graphs} \label{section_line}

The \textit{line graph} of a graph $G$ is denoted by $L(G)$, with vertex set $E(G)$. An edge is drawn between two vertices in $L(G)$ if the corresponding edges in $G$ share a common vertex. In this section, we will prove that the conjecture is true for the line graph of a complete graph $K_{n}$.

\begin{thm}
    If a graph \( G \) is the line graph of complete graph \( K_n \), where \( n > 5 \), then
    \[
    \frac{\binom{n}{2}}{\binom{n}{2} - \sqrt{s^+}} \leq \omega(G).
    \]
\end{thm}

\begin{proof}
The line graph L($K_{n}$) is a strongly regular graph of the form
\[
L(K_{n}) = srg\left( \binom{n}{2},2(n-2),n-2,4\right).
\]

Every vertex of $K_n$ will have $n-1$ edges to other vertices. From the line graph definition, the vertices corresponding to those $n-1$ edges form a clique as they all share a common vertex. This implies that 
\[\omega(L(K_n)) \geq n-1 \]
and the conjecture is true if,
\begin{equation}
    \frac{\binom{n}{2}}{\binom{n}{2}-\sqrt{s^+}} \leq n-1 \leq \omega(L(K_n)).
\end{equation}

As $G$ is a strongly regular graph with parameter $\left( \binom{n}{2},2(n-2),n-2,4\right)$, the positive eigenvalues and corresponding multiplicities of $G$ are,
\[
2(n-2) \text{ with multiplicity } 1,
\]
and if $n>5$, then 
\[
r=n-4 \text{ with multiplicity } f<\frac{n(n-1)}{4}.
\]
From the above eigenvalues the conjecture is true if
\[
r^{2}f \leq \left(\frac{n^2}{4} - 4 \right) (n-2)^2.
\]
We know the value of $r$ and $f$, therefore
\[
r^2f < (n-4)^2\left(\frac{n(n-1)}{4}\right).
\]

Now we aim to prove,
$$(n-4)^2\left(\frac{n(n-1)}{4}\right) \leq \left(\frac{n^2}{4} - 4 \right) (n-2)^2. $$
If $n>4$, then
\begin{align*}
    n(n-4)(n-1) &\leq (n+4)(n-2)^2 \\
    n^3-5n^2+4n &\leq n^3-12n+16 \\
    0 &\leq 5n^2-16n+16.
\end{align*}

Here $5n^2-16n+16$ is always positive for $n>0$.
Hence the conjecture is true for line graphs of complete graph.
\end{proof}





\section{Cartesian product}\label{section_car}
The Cartesian product $G\; \Box \;H$ of graphs $G$ and $H$ is a graph with vertex set $V(G) \times V(H)$, where two vertices $(u, v)$ and $(u', v')$ are adjacent if and only if either (i) $u = u'$ and $(v, v') \in E(H)$, or (ii) $v = v'$ and $(u, u') \in E(G)$. One nice property of Cartesian product is: if $\lambda$ and $\lambda'$ are eigenvalues of $G$ and $H$, respectively, then $\lambda + \lambda'$ is an eigenvalue of $G\; \Box \;H$. In this section, we first present a result on the clique number of $G$ and $G \; \Box \; G$, and then use this to prove the conjecture.

\begin{lemma}
Let \( G \) be any undirected graph and \( G \, \Box \, G \) be the Cartesian product of $G$ with $G$, then \( \omega(G) = \omega( G \, \Box \, G ) \).
\end{lemma}

\begin{proof}
We prove \( \omega(G) = \omega(G \, \Box \, G) \) by showing both inequalities. First, let \( C = \{v_1, v_2, \ldots, v_\omega\} \) be a maximum clique in \( G \). For any fixed \( u \in V(G) \), the set 
\[
C' = \{(u, v_1), (u, v_2), \ldots, (u, v_\omega)\}
\]
forms a clique in \( G \, \Box \, G \). Hence, \( \omega(G \, \Box \, G) \geq \omega(G) \).

Now for the other direction, assume for contradiction that \( \omega(G \, \Box \, G) > \omega(G) \). Let \( C' = \{(u_1, v_1), \ldots, (u_{\omega+1}, v_{\omega+1})\} \) be a clique in \( G \, \Box \, G \). We analyze three cases. Case 1: All \( u_i \) are identical. Then \( \{v_1, \ldots, v_{\omega+1}\} \) forms a clique in \( G \), contradicting \( \omega(G) = \omega \). Case 2: All \( v_i \) are identical. Then \( \{u_1, \ldots, u_{\omega+1}\} \) forms a clique in \( G \), again a contradiction. Case 3: There exist \( (u_i, v_i) \), \( (u_j, v_j) \) in \( C' \) with \( u_i \neq u_j \) and \( v_i \neq v_j \). These two vertices cannot be adjacent in \( G \, \Box \, G \), violating the clique property. Hence, all vertices in \( C' \) must share at least one coordinate, reducing to Case 1 or 2. Since all cases lead to contradictions, \( \omega(G \, \Box \, G) \leq \omega(G) \). Combining both inequalities, we conclude \( \omega(G \, \Box \, G) = \omega(G) \).
\end{proof}

\begin{thm}
    Let graph $G$ be a strongly regular graph with parameters $(n, d, \lambda, \mu)$. If $G$ satisfies the conjecture and $n>7$, then $G \; \Box \; G$ also satisfies the conjecture.
\end{thm}

\begin{proof}

Let $G'=G\; \Box \;G$ and let $d>r>s$ are the eigenvalue of $G$ with multiplicity 1, $f$ and $g$, respectively. let $s^+$ and $t^+$ be the sums of squares of all positive eigenvalues of $G$ and $G'$, respectively. As $s<0$, therefore
\[
s^+=d^2+fr^2
\]
\[
t^+ = 4d^2  +  2f(d+r)^2 + 2g(d+s)^2 + 2fg(r+s)^2 + 4f^2r^2. 
\]
As
\[
s \leq r < d, \:r+s<r
\] 
Therefore
\begin{align*}
    t^+ &\leq 4d^2  +  2f(2d)^2 + 2g(2d)^2 + 2fgr^2 + 4f^2r^2\\
    t^+ &\leq 4d^2  +  2f(2d)^2 + 2g(2d)^2 + 4fgr^2 + 4f^2r^2\\
    t^+ &\leq 4(1+2f+2g)d^2 + 4f(g+f) r^2.
\end{align*}

Now, if $n>7$ then 
\[
4(1+2f+2g)<4(2n)\leq n^2 
\]
and, since $n>7$  
\[
4f(g+f)r^2= 4f(n-1)r^2 < fn^2r^2.
\]
From above we can write 
\begin{align*}
    t^+ &\leq n^2d^2+fn^2r^2\\
    t^+ &\leq n^2(d^2+fr^2)\\
    t^+ &\leq n^2s^+.
\end{align*}
Therefore,
\[
\frac{n^2}{n^2 - \sqrt{t^+}} \leq \frac{n}{n - \sqrt{s^+}} \leq \omega(G) = \omega(G').
\]
Hence proved.
\end{proof}

\section{Ramanujan graph}\label{section_ramanujan}

The Ramanujan graph is the best expander graph, where an expander graph is defined as a $d$-regular graph that is sparse yet well-connected. The second-largest eigenvalue $\lambda_2$ of $A(G)$, the adjacency matrix of graph $G$, is pivotal in analyzing the quality of an expander graph. A smaller $\lambda_2$ indicates a better expander graph. An expander graph is called a Ramanujan graph if $\lambda_2 \leq 2\sqrt{d-1}$. For further reading on Ramanujan graphs, refer to \cite{hoory}.

\begin{theorem}
For any $d$-regular Ramanujan graph $G$, the conjecture
\[
\frac{n}{n-\sqrt{s^+}} \leq \omega(G)
\]
holds true if $n \geq 11d$.
\end{theorem}

\begin{proof}
Given that the second largest eigenvalue of a $d$-regular Ramanujan graph is $\lambda_2 \leq 2\sqrt{d-1}$, it follows that
\[
s^+ \leq d^2 + (n-1)4(d-1).
\]
We require that
\[
s^+ \leq \frac{4}{9}n^2.
\]
The conjecture holds if
\[
d^2 + (n-1)4(d-1) \leq \frac{4}{9}n^2.
\]
This can be rewritten as the quadratic inequality:
\begin{equation}\label{ramanujan}
    0 \leq \frac{4}{9}n^2 - 4(d-1)n - (d - 2)^2.
\end{equation}

Treating this as a quadratic equation in terms of $n$, we find its roots:
\[
n = \frac{9(d-1) \pm 3\sqrt{9(d-1)^2 + (d - 2)^2}}{2}.
\]
Thus, if
\[
n \geq \frac{9(d-1) + 3\sqrt{9(d-1)^2 + (d - 2)^2}}{2},
\]
then inequality (\ref{ramanujan}) holds. Simplifying further, we see that
\[
n \geq \frac{9d + 3\sqrt{9d^2 + d^2}}{2}.
\]
Therefore, if $n \geq 11d$, then the conjecture is true.
\end{proof}

\section{Conclusion}
The conjecture proposes a stronger spectral bound for the clique number. In this paper, we have settled this conjecture for some different classes of graphs, specially strongly regular graphs, but it remains open for general graphs.

\subsection*{Acknowledgments}
The authors would like to thank Sivaramakrishnan Sivasubramanian and Sudipta Mallik for their valuable suggestions on this paper. This work was partially supported by the Department of Science and Technology (Government of India) under the SERB Project (project number SRG/2022/002219), DST INSPIRE program (grant number Inspire 16/2020), and project IITI/YFRCG/2023-24/03.


\end{document}